\documentclass[envcountsect,envcountsame,11pt]{llncs}
\usepackage[a4paper,hmargin=0.98in,vmargin=0.98in]{geometry}

\usepackage{framed,xspace}
\usepackage{bbm,qip}
\usepackage{color,graphicx}
\usepackage{ifpdf}
\usepackage{tikz}

\newcommand{\dP}{{\sf P}^*}
\newcommand{\dhP}{\hat{\sf P}^*}
\newcommand{\A}{{\sf A}}
\newcommand{\dA}{{\sf A}^*}

\newcommand{\dhA}{\hat{\sf A}^*}
\newcommand{\B}{{\sf B}}
\newcommand{\dB}{{\sf B}^*}

\newcommand{\dhB}{\hat{\sf B}'}
\newcommand{\ddhB}{\hat{\sf B}^*}

\newcommand{\CFProtocol}{\mathtt{Coin-Flip \ Protocol}}
\newcommand{\IQZK}{\mathtt{IQZK \ Protocol}}
\newcommand{\NIZKProtocol}{\mathtt{NIZK \ Protocol}}
\newcommand{\IQZKF}{\mathtt{IQZK^{\mathcal{F_{\sf{COIN}}}} \ Protocol}}
\newcommand{\F}{\mathcal{F}}
\newcommand{\cF}{\mathcal{F_{\sf{COIN}}}}

\newcommand{\CFP}{\mathtt{CFP}}
\newcommand{\CRS}{\mathtt{CRS}}
\newcommand{\NIZK}{\mathtt{NIZK}}
\newcommand{\COIN}{\mathtt{COIN}}
\newcommand{\start}{\scriptstyle{\sf{START}}}

\newcommand{\ok}{\scriptstyle{\sf{OK}}}
\newcommand{\refuse}{\scriptstyle{\sf{REFUSE}}}
\newcommand{\fail}{\scriptstyle{\sf{FAIL}}}
\newcommand{\SNIZK}{\sf{\hat{S}_{\mathtt{NIZK}}}}
\newcommand{\SIQZK}{\sf{\hat{S}_{\mathtt{IQZK}}}}
\newcommand{\SIQZKF}{\sf{\hat{S}_{\mathtt{IQZK^{\mathcal{F_{\sf{COIN}}}}}}}}

\newcommand{\GH}{{\cal G}_{\tt H}}
\newcommand{\GB}{{\cal G}_{\tt B}}

\newcommand{\pkB}{\mathtt{pkB}}
\newcommand{\pkH}{\mathtt{pkH}}
\newcommand{\sk}{\mathtt{sk}}

\newcommand{\approxq}{\stackrel{\text{\it\tiny q}}{\approx}}
\newcommand{\approxs}{\stackrel{\text{\it\tiny s}}{\approx}}

\newcommand{\crs}{\omega}
\newcommand{\negl}[1]{\mathit{negl}({#1})}
\newcommand{\set}[1]{\{#1\}}

\newcounter{itm}
\newenvironment{myprotocol}[1]
  {\begin{minipage}{\columnwidth} 
    \begin{framed}\hspace{0ex} 
     \begin{minipage}{0.99\columnwidth}
       {\bf #1:}
       \setcounter{itm}{1}
       \begin{list}{\arabic{itm}.}{\usecounter{itm}}}
   {    \end{list}
       \vspace{-1.5ex} 
       \end{minipage} 
     \end{framed} 
    \end{minipage}\vspace{-0.6ex}}

\newenvironment{myfigure}[1]
         {\begin{figure}[#1] \centering}
         { \end{figure}}

\ifpdf

\else

\fi

\begin{document}
\allowdisplaybreaks
\pagestyle{plain}
\title{Quantum-Secure Coin-Flipping and Applications}

\author{Ivan Damg{\aa}rd \and Carolin Lunemann} 
\institute{DAIMI, Aarhus University, Denmark\\\email{\{ivan|carolin\}@cs.au.dk} }

\maketitle

\begin{abstract}
In this paper, we prove classical coin-flipping secure in the presence of quantum adversaries. The proof uses a recent result of Watrous~\cite{Watrous09} that allows quantum rewinding for protocols of a certain form. We then discuss two applications. First, the combination of coin-flipping with any non-interactive zero-knowledge protocol leads to an easy transformation from non-interactive zero-knowledge to interactive quantum zero-knowledge. Second, we discuss how our protocol can be applied to a recently proposed method for improving the security of quantum protocols~\cite{DFLSS09}, resulting in an implementation without set-up assumptions. Finally, we sketch how to achieve efficient simulation for an extended construction in the common-reference-string model. \\[2ex]
\end{abstract}


\section{Introduction}
\label{sec:intro}

In this paper, we are interested in a standard coin-flipping protocol with classical messages exchange but where the adversary is assumed to be capable of quantum computing. Secure coin-flipping allows two parties Alice and Bob to agree on a uniformly random bit in a fair way, i.e., neither party can influence the value of the coin to his advantage. The (well-known) protocol proceeds as follows: Alice commits to a bit $a$, Bob then sends bit $b$, Alice opens the commitment and the resulting coin is the exclusive disjunction of both bits, i.e.\ $coin = a \oplus b$.

For Alice's commitment to her first message, we assume a classical bit commitment scheme. Intuitively, a commitment scheme allows a player to commit to a value, while keeping it hidden ({\em hiding property}) but preserving the possibility to later reveal the value fixed at commitment time ({\em binding property}). More formally, a bit commitment scheme takes a bit and some randomness as input. The hiding property is formalized by the non-existence of a distinguisher able to distinguish with non-negligible advantage between a commitment to 0 and a commitment to 1. The binding property is fulfilled, if it is infeasible for a forger to open one commitment to both values 0 and 1. The hiding respectively binding property holds with unconditional (i.e.\ perfect or statistical) security in the classical and the quantum setting, if the distinguisher respectively the forger is unrestricted with respect to his (quantum-) computational power. In case of a polynomial-time bounded classical distinguisher respectively forger, the commitment is computationally hiding respectively binding. The computationally hiding property translates to the quantum world by simply allowing the distinguisher to be quantum. However, the case of a quantum forger can not be handled in such a straightforward manner, due to the difficulties of rewinding in general quantum systems (see e.g.~\cite{vandeGraaf97,DFS04,Watrous09} for discussions).

For our basic coin-flip protocol, we assume the commitment to be {\it unconditionally binding} and {\it computationally hiding against a quantum adversary}.\footnote{Recall that unconditionally secure commitments, i.e.\ unconditionally hiding and binding at the same time, are impossible in both the classical and the quantum world.} Thus, we achieve unconditional security against cheating Alice and quantum-computational security against dishonest Bob. Such a commitment scheme follows, for instance, from any pseudorandom generator~\cite{Naor91}, secure against a quantum distinguisher. Even though the underlying computational assumption, on which the security of the embedded commitment is based, withstands quantum attacks, the security proof of the entire protocol and its integration into other applications could previously not be naturally translated from the classical to the quantum world. Typically, security against a classical adversary is argued using rewinding of the adversary. But in general, rewinding as a proof technique cannot be directly applied, if Bob runs a quantum computer: First, the intermediate state of a quantum system cannot be copied~\cite{WZ82}, and second, quantum measurements are in general irreversible. Hence, in order to produce a classical output, the simulator had to (partially) measure the quantum system without copying it beforehand, but then it would become generally impossible to reconstruct all information necessary for correct rewinding. For these reasons, no simple and straightforward security proofs for the quantum case were previously known. 

In this paper, we show the most natural and direct quantum analogue of the classical security proof for standard coin-flipping, by using a recent result of Watrous~\cite{Watrous09}. Watrous showed how to construct an efficient quantum simulator for quantum verifiers for several zero-knowledge proof systems such as graph isomorphism, where the simulation relies on the newly introduced \emph{quantum rewinding theorem}. We now show that his quantum rewinding argument can also be applied to classical coin-flipping in a quantum world. 

By calling the coin-flip functionality sequentially a sufficient number of times, the communicating parties can interactively generate a common random string from scratch. The generation can then be integrated into other (classical or quantum) cryptographic protocols that work in the common-reference-string model. This way, several interesting applications can be implemented entirely in a simple manner without any set-up assumptions. Two example applications are discussed in the second part of the paper. 

The first application relates to zero-knowledge proof systems, an important building block for larger cryptographic protocols. Recently, Hallgren et al.~\cite{HKSZ08} showed that any honest verifier zero-knowledge protocol can be made zero-knowledge against any classical and quantum verifier. Here we show a related result, namely, a simple transformation from non-interactive (quantum) zero-knowledge to interactive quantum zero-knowledge. A non-interactive zero-knowledge proof system can be trivially turned into an interactive {\em honest verifier} zero-knowledge proof system by just letting the verifier choose the reference string. Therefore, this consequence of our result also follows from~\cite{HKSZ08}. However, our proof is much simpler. In general, the difference between us and~\cite{HKSZ08} is that our focus is on establishing coin-flipping as a stand-alone tool that can be used in several contexts rather than being integrated in a zero-knowledge construction as in~\cite{HKSZ08}.

As second application we discuss the interactive generation of a common reference string for the general compiler construction improving the security of a large class of quantum protocols that was recently proposed in~\cite{DFLSS09}. Applying the compiler, it has been shown how to achieve hybrid security in existing protocols for password-based identification~\cite{DFSS07} and oblivious transfer~\cite{BBCS91} without significant efficiency loss, such that an adversary must have both large quantum memory \emph{and} large computing power to break the protocol. Here we show how a common reference string for the compiler can be generated from scratch according to the specific protocol requirements in~\cite{DFLSS09}.

Finally, we sketch an extended commitment scheme for quantum-secure coin-flipping in the common-reference-string model. This construction can be \emph{efficiently} simulated without the need of rewinding, which is necessary to claim universal composability.


\section{Preliminaries}

\subsection{Notation}

We assume the reader's familiarity with basic notation and concepts of quantum information processing as in standard literature, e.g.~\cite{NC00}. Furthermore, we will only give the details of the discussed applications that are most important in the context of this work. A full description of the applications can be found in the referenced papers. 

We denote by $\negl{n}$ any function of $n$, if for any polynomial $p$ it holds that $\negl{n} \leq 1/p(n)$ for large enough $n$. As a measure of \emph{closeness} of two quantum states $\rho$ and $\sigma$, their trace distance $\delta(\rho,\sigma) = \frac{1}{2} \tr(|\rho-\sigma|)$ or square-fidelity $\bra{\rho}\sigma\ket{\rho}$ can be applied. A quantum algorithm consists of a family $\{ C_n\}_{n \in \naturals}$ of quantum circuits and is said to run in polynomial time, if the number of gates of $C_n$ is polynomial in $n$. Two families of quantum states $\set{\rho_n}_{n \in \naturals}$ and $\set{\sigma_n}_{n \in \naturals}$ are called \emph{quantum-computationally indistinguishable}, denoted $\rho \approxq \sigma$, if any polynomial-time quantum algorithm has negligible advantage in $n$ of distinguishing $\rho_n$ from $\sigma_n$. Analogously, they are \emph{statistically indistinguishable}, denoted $\rho \approxs \sigma$, if their trace distance is negligible in~$n$. For the reverse circuit of quantum circuit $Q$, we use the standard notation for the transposed, complex conjugate operation, i.e.\ $Q^\dagger$. The controlled-NOT operation (CNOT) with a control and a target qubit as input flips the target qubit, if the control qubit is 1. In other words, the value of the second qubit corresponds to the classical exclusive disjunction (XOR). A phase-flip operation can be described by Pauli operator $Z$. For quantum state $\rho$ stored in register $\sf R$ we write $\ket{\rho}_R$.

\subsection{Definition of Security} 
\label{sec:security}

We follow the framework for defining security which was introduced in~\cite{FS09} and also used in~\cite{DFLSS09}. Our cryptographic two-party protocols run between player Alice, denoted by $\A$, and player Bob ($\B$). Dishonest parties are indicated by $\dA$ and $\dB$, respectively. The security against a dishonest player is based on the {\it real/ideal-world paradigm} that assumes two different worlds: The {\em real-world} that models the actual protocol $\Pi$ and the {\em ideal-world} based on the ideal functionality $\F$ that describes the intended behavior of the protocol. If both executions are indistinguishable, security of the protocol in real life follows. In other words, a dishonest real-world player $\dP$ that attacks the protocol cannot achieve (significantly) more than an ideal-world adversary $\dhP$ attacking the corresponding ideal functionality.

More formally, the joint input state consists of classical inputs of honest parties and possibly quantum input of dishonest players. A protocol $\Pi$ consists of an infinite family of interactive (quantum) circuits for parties $\A$ and $\B$. A classical (non-reactive) ideal functionality $\F$ is given by a conditional probability distribution $P_{\F(in_A,in_B)|in_A in_B}$, inducing a pair of random variables $(out_A,out_B) = \F(in_A,in_B)$ for every joint distribution of $in_A$ and $in_B$, where $in_P$ and $out_P$ denote party $\sf{P}$'s in- and output, respectively. For the definition of (quantum-) computational security against a dishonest Bob, a polynomial-size (quantum) \emph{input sampler} is considered, which produces the input state of the parties.

\begin{definition}[Correctness] \label{def:correctness}
A protocol $\Pi$ \emph{correctly implements} an ideal classical functionality $\F$, if for every distribution of the input values of honest Alice and Bob, the resulting common outputs of $\Pi$ and $\F$ are \emph{statistically indistinguishable}.
\end{definition}

\begin{definition}[Unconditional security against dishonest Alice] \label{def:unboundedAlice} 
A protocol $\Pi$ implements an ideal classical functionality $\F$ \emph{unconditionally securely} against dishonest Alice, if for any real-world adversary $\dA$, there exists an ideal-world adversary $\dhA$, such that for any input state it holds that the output state, generated by $\dA$ through interaction with honest $\B$ in the real-world, is \emph{statistically indistinguishable} from the output state, generated by $\dhA$ through interaction with $\F$ and $\dA$ in the ideal-world.
\end{definition}

\begin{definition}[(Quantum-) Computational security against dishonest Bob] \label{def:polyboundedBob} 
A protocol $\Pi$ implements an ideal classical functionality $\F$ \emph{(quantum-) computationally securely} against dishonest Bob, if for any (quantum-) computationally bounded real-world adversary $\dB$, there exists a (quantum-) computationally bounded ideal-world adversary $\ddhB$, such that for any efficient input sampler, it holds that the output state, generated by $\dB$ through interaction with honest $\A$ in the real-world, is \emph{(quantum-) computationally indistinguishable} from the output state, generated by $\ddhB$ through interaction with $\F$ and $\dB$ in the ideal-world.
\end{definition}

For more details and a definition of indistinguishability of quantum states, see~\cite{FS09}. There, it has also been shown that protocols satisfying the above definitions compose sequentially in a classical environment. Furthermore, note that in Definition~\ref{def:unboundedAlice}, we do not necessarily require the ideal-world adversary $\dhA$ to be efficient. We show in Section~\ref{sec:CFP-crs} how to extend our coin-flipping construction such that we can achieve an efficient simulator.

The coin-flipping scheme in Section~\ref{sec:CFP-crs} as well as the example applications in Sections~\ref{sec:IQZK} and~\ref{sec:DualProtocols} work in the common-reference-string (CRS) model. In this model, all participants in the real-world protocol have access to a classical public CRS, which is chosen before any interaction starts, according to a distribution only depending on the security parameter. However, the participants in the ideal-world interacting with the ideal functionality do not make use of the CRS. Hence, an ideal-world simulator $\dhP$ that operates by simulating a real-world adversary $\dP$ is free to choose a string in any way he wishes.


\section{Quantum-Secure Coin-Flipping}

\subsection{The Coin-Flip Protocol}
\label{subsec:CFP}

Let $n$ indicate the security parameter of the commitment scheme which underlies the protocol. We use an {\it unconditionally binding} and {\it quantum-computationally hiding} commitment scheme that takes a bit and some randomness $r$ of length $l$ as input, i.e.\ $com: \{ 0,1 \} \times \{ 0,1 \}^l \rightarrow \{ 0,1 \}^{l+1}$. The unconditionally binding property is fulfilled, if it is impossible for any forger to open one commitment to both 0 and 1, i.e.\ to compute $r,r'$ such that $com(0,r) = com(1,r')$. Quantum-computationally hiding is ensured, if no quantum distinguisher can distinguish between $com(0,r)$ and $com(1,r')$ for random $r,r'$ with non-negligible advantage. As mentioned earlier, for a specific instantiation we can use, for instance, Naor's commitment based on a pseudorandom generator~\cite{Naor91}. This scheme does not require any initially shared secret information and is secure against a quantum distinguisher. \footnote{We describe the commitment scheme in this simple notation. However, if it is based on a specific scheme, e.g.~\cite{Naor91}, the precise notation has to be slightly adapted.} 

We let Alice and Bob run the $\CFProtocol$ (see Fig.~\ref{fig:CFP}), which interactively generates a random and fair $coin$ in one execution and does not require any set-up assumptions. Correctness is obvious by inspection of the protocol: If both players are honest, they independently choose random bits. These bits are then combined via exclusive disjunction, resulting in a uniformly random $coin$. 

\begin{figure}
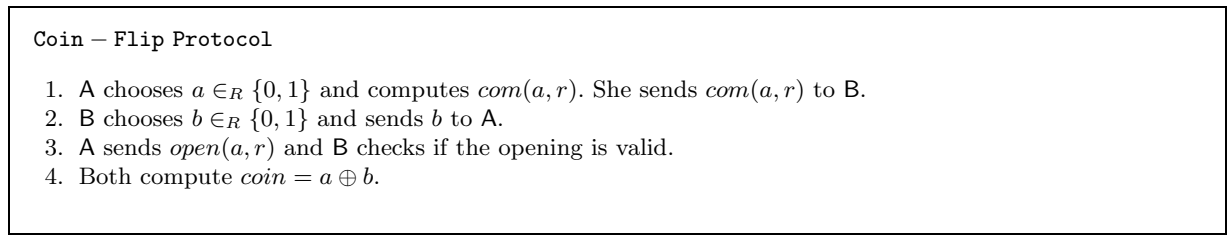

\begin{framed}
$\CFProtocol$
\begin{enumerate}
\item $\A$ chooses $a \in_R \{0,1\}$ and computes $com(a,r)$. She sends $com(a,r)$ to $\B$.
\item $\B$ chooses $b \in_R \{0,1\}$ and sends $b$ to $\A$.
\item $\A$ sends $open(a,r)$ and $\B$ checks if the opening is valid.
\item Both compute $coin = a \oplus b$.
\end{enumerate}
\end{framed}
\vspace{-2ex}
\small
\caption{The Coin-Flip Protocol.}\label{fig:CFP}
\vspace{-1ex}
\end{figure}

The corresponding ideal coin-flip functionality $\cF$ is described in Figure~\ref{fig:cF}. Note that dishonest $\dA$ may refuse to open $com(a,r)$ in the real-world after learning $\B$'s input. For this case, $\cF$ allows her a second input $\refuse$, leading to output $\fail$ and modeling the abort of the protocol.

\begin{figure}
\begin{framed}
$\mathtt{Ideal \ Functionality} \ \cF$:\\[2ex]
Upon receiving requests $\start$ from Alice and Bob, $\cF$ outputs a uniformly random $coin$ to Alice. It then waits to receive Alice's second input $\ok$ or $\refuse$ and outputs $coin$ or $\fail$ to Bob, respectively.
\vspace{-1ex}
\end{framed}
\vspace{-2ex}
\small
\caption{The Ideal Coin-Flip Functionality.}\label{fig:cF}
\vspace{-1ex}
\end{figure}

\subsection{Security}
\label{subsec:SecurityCFP}

\begin{theorem}
The $\CFProtocol$ is \emph{unconditionally secure} against any unbounded dishonest Alice according to Definition~\ref{def:unboundedAlice}, provided that the underlying commitment scheme is unconditionally binding.
\end{theorem}

\begin{proof}
We construct an ideal-world adversary $\dhA$, such that the real output of the protocol is statistically indistinguishable from the ideal output produced by $\dhA$, $\cF$ and $\dA$.

\begin{figure}
\begin{framed}
$\mathtt{Ideal-World \ Simulation \ \ \dhA}$:
\begin{enumerate}
\item Upon receiving $com(a,r)$ from $\dA$, $\dhA$ sends $\start$ and then $\ok$ to $\cF$ as first and second input, respectively, and receives a uniformly random $coin$.
\item\label{step:compute-a} $\dhA$ computes $a$ and $r$ from $com(a,r)$. 
\item\label{step:compute-b} $\dhA$ computes $b = coin \oplus a$ and sends $b$ to $\dA$.
\item $\dhA$ waits to receive $\dA$'s last message and outputs whatever $\dA$ outputs.
\end{enumerate}
\vspace{-1ex}
\end{framed}
\vspace{-2ex}
\small
\caption{The Ideal-World Simulation $\dhA$.}\label{fig:simulationA}
\vspace{-1ex}
\end{figure}

First note that $a, r$ and $com(a,r)$ are chosen and computed as in the real protocol. From the statistically binding property of the commitment scheme, it follows that $\dA$'s choice bit $a$ is uniquely determined from $com(a,r)$, since for any $com$, there exists at most one pair $(a,r)$ such that $com = com(a,r)$ (except with probability negligible in $n$). Hence in the real-world, $\dA$ is unconditionally bound to her bit before she learns $\B$'s choice bit, which means $a$ is independent of $b$. Therefore in Step~\ref{step:compute-a}, the simulator can correctly (but not necessarily efficiently) compute $a$ (and $r$). Note that, in the case of unconditional security, we do not have to require the simulation to be efficient. We show in Section~\ref{sec:CFP-crs} how to extend the commitment in order to extract $\dA$'s inputs efficiently. Finally, due to the properties of XOR, $\dA$ cannot tell the difference between the random $b$ computed (from the ideal, random $coin$) in the simulation in Step~\ref{step:compute-b} and the randomly chosen $b$ of the real-world. It follows that the simulated output is statistically indistinguishable from the output in the real protocol.
\qed
\end{proof}

To prove security against any dishonest quantum-computationally bounded $\dB$, we show that there exists an ideal-world simulation $\ddhB$ with output quantum-computationally indistinguishable from the output of the protocol in the real-world. In a classical simulation, where we can simply use rewinding, a polynomial-time simulator works as follows. It inquires $coin$ from $\cF$, chooses random $a$ and $r$, and computes $b' = coin \oplus a$ as well as $com(a,r)$. It then sends $com(a,r)$ to $\dB$ and receives $\dB$'s choice bit $b$. If $b = b'$, the simulation was successful. Otherwise, the simulator rewinds $\dB$ and repeats the simulation. Note that our security proof should hold also against any quantum adversary. The polynomial-time quantum simulator proceeds similarly to its classical analogue but requires quantum registers as work space and relies on the {\it quantum rewinding lemma} of Watrous~\cite{Watrous09} (see Lemma~\ref{lem:qrewind}).

In the paper, Watrous proves how to construct a quantum zero-knowledge proof system for graph isomorphism using his (ideal) quantum rewinding lemma. The protocol proceeds as a $\Sigma$-protocol, i.e.\ a protocol in three-move form, where the verifier flips a single coin in the second step and sends this challenge to the prover. Since these are the essential aspects also in our $\CFProtocol$, we can apply Watrous' quantum rewinding technique (with slight modifications) as a black-box to our protocol. We also follow his notation and line of argument here. For a more detailed description and proofs, we refer to~\cite{Watrous09}.

\begin{theorem}
The $\CFProtocol$ is \emph{quantum-computationally secure} against any \\ polynomial-time bounded, dishonest Bob according to Definition~\ref{def:polyboundedBob}, provided that the underlying commitment scheme is quantum-computationally hiding and the success probability of quantum rewinding achieves a non-negligible lower bound $p_0$.
\end{theorem}

\begin{proof}
Let $W$ denote $\dB$'s auxiliary input register, containing an $\tilde{n}$-qubit state $\ket{\psi}$. Furthermore, let $V$ and $B$ denote $\dB$'s work space, where $V$ is an arbitrary polynomial-size register and $B$ is a single qubit register. $\A$'s classical messages are considered in the following as being stored in quantum registers $A_1$ and $A_2$. In addition, the quantum simulator uses registers $R$, containing all possible choices of a classical simulator, and $G$, representing its guess $b'$ on $\dB$'s message $b$ in the second step. Finally, let $X$ denote a working register of size $\tilde{k}$, which is initialized to the state $\ket{0^{\tilde{k}}}$ and corresponds to the collection of all registers as described above except $W$.

The quantum rewinding procedure is implemented by a general quantum circuit $R_{coin}$ with input $(W,X, \dB, coin)$. As a first step, it applies a unitary $(\tilde{n},\tilde{k})$-quantum circuit $Q$ to $(W,X)$ to simulate the conversation, obtaining registers $(G,Y)$. Then, a test takes place to observe whether the simulation was successful. In that case, $R_{coin}$ outputs the resulting quantum register. Otherwise, it {\it quantumly rewinds} by applying the reverse circuit $Q^\dag$ on $(G,Y)$ to retrieve $(W,X)$ and then a phase-flip transformation on $X$ before another iteration of $Q$ is applied. Note that $R_{coin}$ is essentially the same circuit as $R$ described in~\cite{Watrous09}, but in our application it depends on the value of a given $coin$, i.e., we apply $R_0$ or $R_1$ for $coin = 0$ or $coin = 1$, respectively. In more detail, $Q$ transforms $(W,X)$ to $(G,Y)$ by the following unitary operations:
\begin{enumerate}
\item[(1)] It first constructs the superposition
$$\frac{1}{\sqrt{2^{l+1}}} \sum_{a,r} \ket{a,r}_{R} \ket{com(a,r)}_{A_1} \ket{b' = coin \oplus a}_{G} \ket{open(a,r)}_{A_2} \ket{0}_{B} \ket{0^{\tilde{k'}}}_{V} \ket{\psi}_{W} \, ,$$
where $\tilde{k'} < \tilde{k}$. Note that the state of registers $(A_1, G, A_2)$ corresponds to a uniform distribution of possible transcripts of the interaction between the players.
\item[(2)] For each possible $com(a,r)$, it then simulates $\dB$'s possible actions by applying a unitary operator to $(W,V,B,A_1)$ with $A_1$ as control:
$$\frac{1}{\sqrt{2^{l+1}}} \sum_{a,r} \ket{a,r}_{R} \ket{com(a,r)}_{A_1} \ket{b'}_{G} \ket{open(a,r)}_{A_2} \ket{b}_{B} \ket{\tilde{\phi}}_{V} \ket{\tilde{\psi}}_{W} \, ,$$
where ${\tilde{\phi}}$ and ${\tilde{\psi}}$ describe modified quantum states.
\item[(3)] Finally, a CNOT-operation is applied to pair $(B,G)$ with $B$ as control to check whether the simulator's guess of $\dB$'s choice was correct. The result of the CNOT-operation is stored in register $G$.
$$\frac{1}{\sqrt{2^{l+1}}} \sum_{a,r} \ket{a,r}_{R} \ket{com(a,r)}_{A_1} \ket{b' \oplus b}_{G} \ket{open(a,r)}_{A_2} \ket{b}_{B} \ket{\tilde{\phi}}_{V} \ket{\tilde{\psi}}_{W} \, .$$
\end{enumerate}
If we denote with $Y$ the register that contains the residual $\tilde{n}+\tilde{k}-1$ -qubit state, the transformation from $(W,X)$ to $(G,Y)$ by applying $Q$ can be written as
$$Q \left( \ket{\psi}_{W} \ket{0^{\tilde{k}}}_{X} \right) = \sqrt{p} \ket{0}_{G} \ket{\phi_{good}(\psi)}_{Y} + \sqrt{1-p} \ket{1}_{G} \ket{\phi_{bad}(\psi)}_{Y}\, ,$$
where $0 < p < 1$ and $\ket{\phi_{good}(\psi)}$ denotes the state, we want the system to be in for a successful simulation. $R_{coin}$ then measures the qubit in register $G$ with respect to the standard basis, which indicates success or failure of the simulation. A successful execution (where $b = b'$) results in outcome 0 with probability $p$. In that case, $R_{coin}$ outputs $Y$. A measurement outcome 1 indicates $b \neq b'$, in which case $R_{coin}$ quantumly rewinds the system, applies a phase-flip (on register $X$) and repeats the simulation, i.e.
$$Q \bigg( 2 \Big( \mathbb{I} \otimes \ket{0^{\tilde{k}}}\bra{0^{\tilde{k}}} \Big) - \mathbb{I} \bigg) Q^\dag \, .$$

Watrous' ideal quantum rewinding lemma (without perturbations) then states the following: Under the condition that the probability $p$ of a successful simulation is non-negligible and independent of any auxiliary input, the output $\rho(\psi)$ of $R$ has square-fidelity close to 1 with state $\ket{\phi_{good}(\psi)}$ of a successful simulation, i.e.,
$$\bra{\phi_{good}(\psi)}\rho(\psi)\ket{\phi_{good}(\psi)} \geq 1 - \varepsilon$$
with error bound $0 < \varepsilon < \frac{1}{2}$. Note that for the special case where $p$ equals $1/2$ and is independent of $\ket{\psi}$, the simulation terminates after at most one rewinding. 

However, we cannot apply the exact version of Watrous' rewinding lemma in our simulation, since the commitment scheme in the protocol is only (quantum-) computationally hiding. Instead, we must allow for small perturbations in the quantum rewinding procedure as follows. Let $adv$ denote $\dB$'s advantage over a random guess on the committed value due to his computing power, i.e.\ $adv = |p - 1/2|$. From the hiding property, it follows that $adv$ is negligible in the security parameter $n$. Thus, we can argue that the success probability $p$ is {\it close} to independent of the auxiliary input and Watrous' quantum rewinding lemma with small perturbations, as stated below (Lemma~\ref{lem:qrewind}) applies with $q = \frac{1}{2}$ and $\varepsilon = adv$. 

\begin{lemma}[Quantum Rewinding Lemma with small perturbations~\cite{Watrous09}]
\label{lem:qrewind}
Let $Q$ be the unitary $(\tilde{n},\tilde{k})$-quantum circuit as given
in~\cite{Watrous09}. Furthermore, let $p_0, q \in (0,1)$ and $\varepsilon \in (0,\frac{1}{2})$ be real numbers such that
	\begin{enumerate}
	\item $|p - q| < \varepsilon$
	\item $p_0(1-p_0) \leq q(1-q)$, and
	\item $p_0 \leq p$ 
	\end{enumerate}
for all $\tilde{n}$-qubit states $\ket{\psi}$. Then there exists a general quantum circuit $R$ of size
$$O \left( \frac{log(1/\varepsilon) size(Q)}{p_0(1-p_0)} \right)$$
such that, for every $\tilde{n}$-qubit state $\ket{\psi}$, the output $\rho(\psi)$ of $R$ satisfies
$$\bra{\phi_{good}(\psi)}\rho(\psi)\ket{\phi_{good}(\psi)} \geq 1 - \varepsilon'$$
where $\varepsilon' = 16\varepsilon \frac{log^2(1/\varepsilon)}{p_0^2 (1-p_0)^2}\ $.
\end{lemma}

Note that all operations in $Q$ can be performed by polynomial-size circuits, and thus, the simulator has polynomial size (in the worst case). Furthermore, $p_0$ denotes the lower bound on the success probability $p$, for which the procedure guarantees correctness. For negligible $\varepsilon$ but non-negligible $p_0$, it follows that $\varepsilon'$ is negligible, and hence, the ``closeness'' of output $\rho(\psi)$ with good state $\ket{\phi_{good}(\psi)}$ is slightly reduced but quantum rewinding remains possible. For a more detailed description of the lemma and the corresponding proofs, we refer to~\cite{Watrous09}.

Finally, to proof security against quantum $\dB$, we construct an ideal-world quantum simulator $\ddhB$ (see Fig.~\ref{fig:simulationB}), interacting with $\dB$ and the ideal functionality $\cF$ and executing Watrous' quantum rewinding algorithm. We then compare the output states of the real process and the ideal process. In case of indistinguishable outputs, quantum-computational security against $\dB$ follows.

\begin{figure}
\begin{framed}
$\mathtt{Ideal-World \ Simulation \ \ \ddhB}$:
\begin{enumerate}
\item $\ddhB$ gets $\dB$'s auxiliary quantum input $W$ and working registers $X$.
\item\label{step:get-coin} $\ddhB$ sends $\start$ and then $\ok$ to $\cF$. It receives a uniformly random $coin$.
\item Depending on the value of $coin$, $\ddhB$ applies the corresponding circuit $R_{coin}$ with input $W, X ,\dB$ and $coin$.
\item $\ddhB$ receives output register $Y$ with $\ket{\phi_{good}(\psi)}$ and ``measures the conversation'' to retrieve the corresponding $(com(a,r),b,open(a,r))$. It outputs whatever $\dB$ outputs.
\end{enumerate}
\vspace{-1ex}
\end{framed}
\vspace{-2ex}
\small
\caption{The Ideal-World Simulation $\ddhB$.}\label{fig:simulationB}
\vspace{-1ex}
\end{figure}

First note that the superposition constructed as described above in circuit $Q$ as Step~(1) corresponds to all possible random choices of values in the real protocol. Furthermore, the circuit models any possible strategy of quantum $\dB$ in Step~(2), depending on control register $\ket{com(a,r)}_{A_1}$. The CNOT-operation on $(B,G)$ in Step~(3), followed by a standard measurement of $G$, indicate whether the guess $b'$ on $\dB$'s choice $b$ was correct. If that was not the case (i.e.\ $b \neq b'$ and measurement result 1), the system gets quantumly rewound by applying reverse transformations (3)-(1), followed by a phase-flip operation. The procedure is repeated until the measurement outcome is 0 and hence $b=b'$. Watrous' technique then guarantees that, assuming negligible $\varepsilon$ and non-negligible $p_0$, then $\varepsilon'$ is negligible and thus, the final output $\rho(\psi)$ of the simulation is close to good state $\ket{\phi_{good}(\psi)}$. It follows that the output of the ideal simulation is indistinguishable from the output in the real-world for any quantum-computationally bounded $\dB$.
\qed
\end{proof}


\section{Applications}
\label{sec:applications}

\subsection{Interactive Quantum Zero-Knowledge}
\label{sec:IQZK}

Zero-knowledge proofs are an important building block for larger cryptographic protocols. The notion of (interactive) zero-knowledge (ZK) was introduced by Goldwasser et al.~\cite{GMR85}. Informally, ZK proofs for any NP language $L$ yield no other knowledge to the verifier than the validity of the assertion proved, i.e.\ $x \in L$. Thus, only this one bit of knowledge is communicated from prover to verifier and zero additional knowledge. For a survey about zero-knowledge, see for instance~\cite{Goldreich01,Goldreich02}.

Blum et al.~\cite{BFM88} showed that the interaction between prover and verifier in any ZK proof can be replaced by sharing a short, random common reference string according to some distribution and available to all parties from the start of the protocol. Note that a CRS is a weaker requirement than interaction. Since all information is communicated mono-directional from prover to verifier, we do not have to require any restriction on the verifier.

As in the classical case, where ZK protocols exist if one-way functions exist, quantum zero-knowledge (QZK) is possible under the assumption that quantum one-way functions exist. In~\cite{K03}, Kobayashi showed that a common reference string or shared entanglement is necessary for non-interactive quantum zero-knowledge. Interactive quantum zero-knowledge protocols in restricted settings were proposed by Watrous in the honest verifier setting~\cite{Watrous02} and by Damg{\aa}rd et al.\ in the CRS model~\cite{DFS04}, where the latter introduced the first $\Sigma$-protocols for QZK withstanding even active quantum attacks. In~\cite{Watrous09}, Watrous then proved that several interactive protocols are zero-knowledge against general quantum attacks.

Recently, Hallgren et al.~\cite{HKSZ08} showed how to transform a $\Sigma$-protocol with stage-by-stage honest verifier zero-knowledge into a new $\Sigma$-protocol that is zero-knowledge against all classical and quantum verifiers. They propose special bit commitment schemes to limit the number of rounds, and view each round as a stage in which an honest verifier simulator is assumed. Then, by using a technique of~\cite{DGW94}, each stage can be converted to obtain zero-knowledge against any classical verifier. Finally, Watrous' quantum rewinding lemma is applied in each stage to prove zero-knowledge also against any quantum verifier.

Here, we propose a simpler transformation from non-interactive (quantum) zero-knowledge (NIZK) to interactive quantum zero-knowledge (IQZK) by combining the $\CFProtocol$ with any $\NIZKProtocol$. Our coin-flipping generates a truly random $coin$ even in the case of a malicious quantum verifier. A sequence of such coins can then be used in any subsequent $\NIZKProtocol$, which is also secure against quantum verifiers, due to its mono-direction. Here, we define a ($\NIZK$)-subprotocol as given in~\cite{BFM88}: Both parties $\A$ and $\B$ get common input $x$. A common reference string $\crs$ of size $k$ allows the prover $\A$, who knows a witness $w$, to give a non-interactive zero-knowledge proof $\pi(\crs,x)$ to a (quantum-) computationally bounded verifier $\B$. By definition, the ($\NIZK$)-subprotocol is complete and sound and satisfies zero-knowledge.

The $\IQZK$ is shown in Figure~\ref{fig:iqzk}. To prove that it is an interactive quantum zero-knowledge protocol, we first construct an intermediate $\IQZKF$ (see Fig.~\ref{fig:iqzkf}) that runs with the ideal functionality $\cF$. Then we prove that the $\IQZKF$ satisfies completeness, soundness and zero-knowledge according to standard definitions. Finally, by replacing the calls to $\cF$ with our $\CFProtocol$, we can complete the transformation to the final $\IQZK$.\\

\begin{myfigure}{h}
\begin{myprotocol}{$\IQZKF$}
\item[$(\COIN)$]
\item\label{step:coin} $\A$ and $\B$ invoke $\cF$ $k$ times. If $\A$ blocks any output $coin_i$ for $i = 1,\ldots,k$ (by sending $\refuse$ as second input), $\B$ aborts the protocol.
\item[$(\CRS)$]
\item\label{step:crs} $\A$ and $\B$ compute $\crs = coin_1 \ldots coin_k$.\medskip
\item[$(\NIZK)$]
\item $\A$ sends $\pi(\crs,x)$ to $\B$. $\B$ checks the proof and accepts or rejects accordingly.
\end{myprotocol}
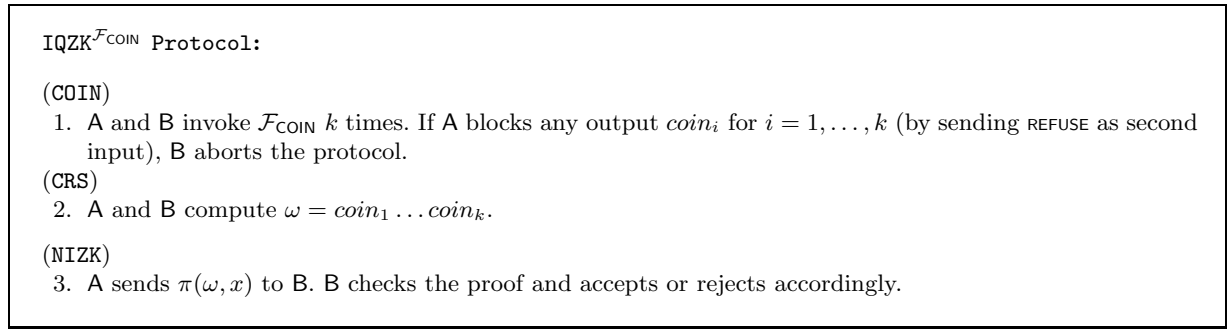
\caption{Intermediate Protocol for IQZK.}\label{fig:iqzkf}
\end{myfigure}

\noindent\emph{\textbf{Completeness:} If $x \in L$, the probability that $(\A,\B)$ rejects $x$ is negligible in the length of $x$.}

From the ideal functionality $\cF$ it follows that each $coin_i$ in Step~\ref{step:coin} is uniformly random for all $i = 1,\ldots,k$. Hence, $\crs$ in Step~\ref{step:crs} is a uniformly random common reference string of size $k$. By definition of any ($\NIZK$)-subprotocol, we have acceptance probability
$$Pr[\crs \in_R \{ 0,1 \}^k, \pi(\crs,x) \leftarrow A(\crs,x,w): B(\crs,x,\pi(\crs,x)) =1] > 1 - \varepsilon'' , $$ 
where $\varepsilon''$ is negligible in the length of $x$. Thus, completeness for the $\IQZKF$ follows.\\

\noindent\emph{\textbf{Soundness:} If $x \notin L$, then for any unbounded prover $\dA$, the probability that $(\dA,\B)$ accepts $x$ is negligible in the length of $x$.}

Any dishonest $\dA$ might stop the $\IQZKF$ at any point during execution. For example, she can block the output in Step~\ref{step:coin} or she can refuse to send a proof $\pi$ in the ($\NIZK$)-subprotocol. Furthermore, $\dA$ can use an invalid $\crs$ (or $x$) for $\pi$. In all of these cases, $\B$ will abort without even checking the proof. Therefore, $\dA$'s best strategy is to ``play the entire game'', i.e.\ to execute the entire $\IQZKF$ without making obvious cheats.

$\dA$ can only convince $\B$ in the ($\NIZK$)-subprotocol of a $\pi$ for any given (i.e.\ normally generated) $\crs$ with negligible probability 
$$Pr[\crs \in_R \{ 0,1 \}^k, \pi(\crs,x) \leftarrow A^*(\crs,x): B(\crs,x,\pi(\crs,x)) =1] \ . $$
Therefore, the probability that $\dA$ can convince $\B$ in the entire $\IQZKF$ in case of $x \notin L$ is also negligible (in the length of $x$) and its soundness follows.\\

\noindent\emph{\textbf{Zero-Knowledge:} An interactive proof system $(\A,\dB)$ for language $L$ is quantum zero-knowledge, if for any quantum verifier $\dB$, there exists a simulator $\SIQZKF$, such that $\SIQZKF \approxq (\A,\dB)$ on common input $x \in L$ and arbitrary additional (quantum) input to $\dB$.}

We construct simulator $\SIQZKF$, interacting with dishonest $\dB$ and simulator $\SNIZK$. Under the assumption on the zero-knowledge property of any $\NIZKProtocol$, there exists a simulator $\SNIZK$ that, on input $x \in L$, generates a randomly looking $\crs$ together with a valid proof $\pi$ for $x$ (without knowing witness $w$). $\SIQZKF$ is described in Figure~\ref{fig:simulationZKF}. It receives a random string $\omega$ from $\SNIZK$, which now replaces the string of coins produced by the calls to $\cF$ in the $\IQZKF$. The ``merging'' of coins into $\crs$ in Step~\ref{step:crs} of the protocol (Fig.~\ref{fig:iqzkf}) is equivalent to the ``splitting'' of $\crs$ into coins in Step~\ref{step:crs-sim} of the simulation (Fig.~\ref{fig:simulationZKF}). Thus, the simulated proof $\pi(\crs,x)$ is indistinguishable from a real proof, which shows that the $\IQZKF$ is zero-knowledge.\\

\begin{figure}
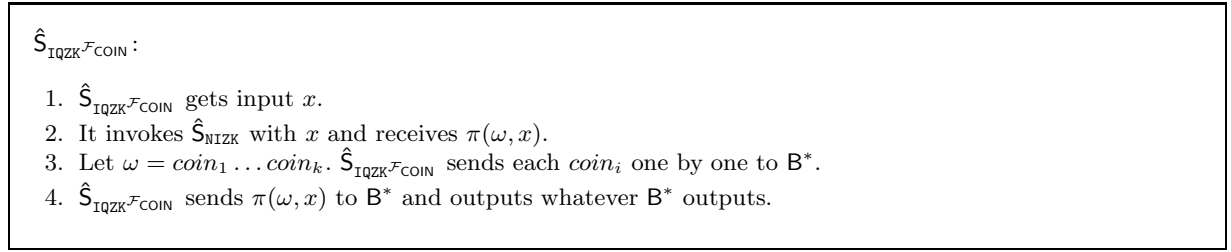

\begin{framed}
$\SIQZKF$:
\begin{enumerate}
\item $\SIQZKF$ gets input $x$.
\item\label{step:nizk-sim} It invokes $\SNIZK$ with $x$ and receives $\pi(\crs, x)$.
\item\label{step:crs-sim} Let $\crs = coin_1 \ldots coin_k$. $\SIQZKF$ sends each $coin_i$ one by one to $\dB$.
\item $\SIQZKF$ sends $\pi(\crs, x)$ to $\dB$ and outputs whatever $\dB$ outputs. 
\end{enumerate}
\vspace{-1ex}
\end{framed}
\vspace{-2ex}
\small
\caption{The Simulation of the Intermediate Protocol for IQZK.}\label{fig:simulationZKF}
\vspace{-1ex}
\end{figure}
\medskip
\begin{myfigure}{h}
\begin{myprotocol}{$\IQZK$}
\item[($\CFP$)] For all $i = 1,\ldots,k$ repeat Steps 1. -- 4.
\item $\A$ chooses $a_i \in_R \{0,1\}$ and computes $com(a_i,r_i)$. She sends $com(a_i,r_i)$ to $\B$.
\item $\B$ chooses $b_i \in_R \{0,1\}$ and sends $b_i$ to $\A$.
\item $\A$ sends $open(a_i,r_i)$ and $\B$ checks if the opening is valid.
\item Both compute $coin_i = a_i \oplus b_i$.\medskip
\item[$(\CRS)$]
\item $\A$ and $\B$ compute $\crs = coin_1 \ldots coin_k$.\medskip
\item[$(\NIZK)$]
\item $\A$ sends $\pi(\crs,x)$ to $\B$. $\B$ checks the proof and accepts or rejects accordingly.
\end{myprotocol}
\caption{Interactive Quantum Zero-Knowledge.}\label{fig:iqzk}
\end{myfigure}

It would be natural to think that the $\IQZK$ could be proved secure simply by showing that the $\IQZKF$ implements some appropriate functionality and then use the composition theorem from~\cite{FS09}. Unfortunately, a zero-knowledge protocol -- which is not necessarily a proof of knowledge -- cannot be modeled by a functionality in a natural way. We therefore instead prove explicitly that the $\IQZK$ has the standard properties of a zero-knowledge proof as follows.\\

\noindent\emph{\textbf{Completeness:}}
From the analysis of the $\CFProtocol$ and its indistinguishability from the ideal functionality $\cF$, it follows that if both players honestly choose random bits, each $coin_i$ for all $i = 1,\ldots,k$ in the ($\CFP$)-subprotocol is generated uniformly at random. Thus, $\crs$ is a random common reference string of size $k$ and the acceptance probability of the ($\NIZK$)-subprotocol as given above holds. Completeness for the $\IQZK$ follows.\\

\noindent\emph{\textbf{Soundness:}}
Again, we only consider the case where $\dA$ executes the entire protocol without making obvious cheats, since otherwise, $\B$ immediately aborts. Assume that $\dA$ could cheat in the $\IQZK$, i.e., $\B$ would accept an invalid proof with non-negligible probability. Then we could combine $\dA$ with simulator $\dhA$ of the $\CFProtocol$ (Fig.~\ref{fig:simulationA}) to show that the $\IQZKF$ was not sound. This, however, is inconsistent with the previously given soundness argument and thus proves by contradiction that the $\IQZK$ is sound.\\

\noindent\emph{\textbf{Zero-Knowledge:}} 
A simulator $\SIQZK$ can be composed of simulator $\SIQZKF$ (Fig.~\ref{fig:simulationZKF}) and simulator $\ddhB$ for the $\CFProtocol$ (Fig.~\ref{fig:simulationB}). $\SIQZK$ gets classical input $x$ as well as quantum input $W$ and $X$. It then receives a valid proof $\pi$ and a random string $\omega$ from $\SNIZK$. As in $\SIQZKF$, $\omega$ is split into $coin_1 \ldots coin_k$. For each $coin_i$, it will then invoke $\ddhB$ to simulate one coin-flip execution with $coin_i$ as result. In other words, whenever $\ddhB$ asks $\cF$ to output a bit (Step~\ref{step:get-coin}, Fig.~\ref{fig:simulationB}), it instead receives this $coin_i$. The transcript of the simulation, i.e.\ $\pi(\crs,x)$ as well as $(com(a_i,r_i), b_i, open(a_i, r_i))$ $\forall i = 1, \ldots, k$ and $\crs = coin_1 \ldots coin_k$, is indistinguishable from the transcript of the $\IQZK$ for any quantum-computationally bounded $\dB$, which concludes the zero-knowledge proof.

\subsection{Generating Commitment Keys for Improved Quantum Protocols}
\label{sec:DualProtocols}

Recently, Damg{\aa}rd et al.~\cite{DFLSS09} proposed a general compiler for improving the security of a large class of quantum protocols. Alice starts such protocols by transmitting random BB84-qubits to Bob who measures them in random bases. Then some classical messages are exchanged to accomplish different cryptographic tasks. The original protocols are typically unconditionally secure against cheating Alice, and secure against a so-called \emph{benignly} dishonest Bob, i.e., Bob is assumed to handle most of the received qubits as he is supposed to. Later on in the protocol, he can deviate arbitrarily. The improved protocols are then secure against an arbitrary computationally bounded (quantum) adversary. The compilation also preserves security in the bounded-quantum-storage model (BQSM) that assumes the quantum storage of the adversary to be of limited size. If the original protocol was BQSM-secure, the improved protocol achieves hybrid security, i.e., it can only be broken by an adversary who has large quantum memory {\it and} large computing power. 

Briefly, the argument for computational security proceeds along the following lines. After the initial qubit transmission from $\A$ to $\B$, $\B$ commits to all his measurement bases and outcomes. The (keyed) dual-mode commitment scheme that is used must have the special properties that the key can be generated by one of two possible key-generation algorithms: $\GH$ or $\GB$. Depending of the key in use, the scheme provides both flavors of security. Namely, with key $\pkH$ generated by $\GH$, respectively $\pkB$ produced by $\GB$, the commitment scheme is unconditionally hiding respectively unconditionally binding. Furthermore, the scheme is secure against a quantum adversary and it holds that $\pkH \approxq \pkB$. The commitment construction is described in full detail in \cite{DFLSS09}.

In the real-life protocol, $\B$ uses the unconditionally hiding key $\pkH$ to maintain unconditional security against any unbounded $\dA$. To argue security against a computationally bounded $\dB$, an information-theoretic argument involving simulator $\dhB$ (see~\cite{DFLSS09}) is given to prove that $\dB$ cannot cheat with the unconditionally binding key $\pkB$. Security in real life then follows from the quantum-computational indistinguishability of $\pkH$ and $\pkB$.

The CRS model is assumed to achieve high efficiency and practicability. Here, we discuss integrating the generation of a common reference string from scratch based on our quantum-secure coin-flipping. Thus, we can implement the {\em entire process} in the quantum world, starting with the generation of a CRS without any initially shared information and using it during compilation as commitment key.\footnote{Note that implementing the entire process comes at the cost of a non constant-round construction, added to otherwise very efficient protocols under the CRS-assumption.}

As mentioned in~\cite{DFLSS09}, a dual-mode commitment scheme can be constructed from the lattice-based cryptosystem of Regev~\cite{Regev05}. It is based on the learning with error problem, which can be reduced from worst-case (quantum) hardness of the (general) shortest vector problem. Hence, breaking Regev's cryptosystem implies an efficient algorithm for approximating the lattice problem, which is assumed to be hard even quantumly. Briefly, the cryptosystem uses dimension $k$ as security parameter and is parametrized by two integers $m$ and $p$, where $p$ is a prime, and a probability distribution on $\mathbb{Z}_p$. A regular public key for Regev's scheme is indistinguishable from a case where a public key is chosen independently from the secret key, and in this case, the ciphertext carries essentially no information about the message. Thus, the public key of a regular key pair can be used as the unconditional binding key $\pkB'$ in the commitment scheme for the ideal-world simulation. Then for the real protocol, an unconditionally hiding commitment key $\pkH'$ can simply be constructed by uniformly choosing numbers in $\mathbb{Z}_p^k \times \mathbb{Z}_p$. Both public keys will be of size $O(mk \log p)$, and the encryption process involves only modular additions, which makes its use simple and efficient.

The idea is now the following. We add (at least) $k$ executions of our $\CFProtocol$ as a first step to the construction of~\cite{DFLSS09} to generate a uniformly random sequence $coin_1 \ldots coin_k$. These $k$ random bits produce a $\pkH'$ as sampled by $\mathcal{G}_\mathtt{H}$, except with negligible probability. Hence, in the real-world, Bob can use $coin_1 \ldots coin_k = \pkH'$ as key for committing to all his basis choices and measurement outcomes. Since an ideal-world adversary $\dhB$ is free to choose any key, it can generate $(\pkB', \sk')$, i.e.\ a regular public key together with a secret key according to Regev's cryptosystem. For the security proof, write $\pkB' = coin_1 \ldots coin_k$. In the simulation, $\dhB$ first invokes $\ddhB$ for each $coin_i$ to simulate one coin-flip execution with $coin_i$ as result. As before, whenever $\ddhB$ asks $\cF$ to output a bit, it instead receives this $coin_i$. Then $\dhB$ has the possibility to decrypt dishonest $\dB$'s commitments during simulation, which binds $\dB$ unconditionally to his committed measurement bases and outcomes. Finally, as we proved in the analysis of the $\CFProtocol$ that $\pkH'$ is a uniformly random string, Regev's proof of semantic security shows that $\pkH' \approxq \pkB'$, and (quantum-) computational security of the real protocols in~\cite{DFLSS09} follows.


\section{On Efficient Simulation in the CRS Model}
\label{sec:CFP-crs}

For our $\CFProtocol$ in the plain model, we cannot claim universal composability. As already mentioned, in case of unconditional security against dishonest $\dA$ according to Definition~\ref{def:unboundedAlice}, we do not require the simulator to be efficient. In order to achieve efficient simulation, $\dhA$ must be able to extract the choice bit efficiently out of $\dA$'s commitment, such that $\dA$'s input is defined after this step. The standard approach to do this is to give the simulator some trapdoor information related to the common reference string, that $\dA$ does not have in real life. Therefore, we extend the commitment scheme to build in such a trapdoor and ensure efficient extraction. To further guarantee UC-security, we circumvent the necessity of rewinding $\dB$ by extending the construction also with respect to equivocability.

We will adapt an approach to our set-up, which is based on the idea of UC-commitments \cite{CF01} and already discussed in the full version of \cite{DFLSS09}. We require a $\Sigma$-protocol for a (quantumly) hard relation $R = \{(x,w)\}$, i.e.\ an honest verifier perfect zero-knowledge interactive proof of knowledge, where the prover shows that he knows a witness $w$ such that the problem instance $x$ is in the language $L$ ($(x,w) \in R$). Conversations are of form $(a_{\Sigma}, c_{\Sigma}, z_{\Sigma})$, where the prover sends $a_{\Sigma}$, the verifier challenges him with bit $c_{\Sigma}$, and the prover replies with $z_{\Sigma}$. For practical candidates of $R$, see e.g.~\cite{DFS04}. Instead of the simple commitment scheme, we use the keyed dual-mode commitment scheme described in Section~\ref{sec:DualProtocols} but now based on a multi-bit version of Regev's scheme \cite{PVW08}. Still we construct it such that depending of the key $\pkH$ or $\pkB$, the scheme provides both flavors of security and it holds that $\pkH \approxq \pkB$.

In real life, the CRS consists of commitment key $\pkB$ and an instance $x'$ for which it holds that $\nexists \ w'$ such that $(x',w') \in R$, where we assume that $x \approxq x'$. To commit to bit $a$, $\A$ runs the honest verifier simulator to get a conversation $(a_{\Sigma}, a, z_{\Sigma})$. She then sends $a_{\Sigma}$ and two commitments $c_0, c_1$ to $\B$, where $c_a = com_{\pkB}(z_{\Sigma},r)$ and $c_{1-a} = com_{\pkB}(0^{z'},r')$ with randomness $r,r'$ and $z' = |z|$. Then, $a, z_{\Sigma}, r$ is send to open the relevant one of $c_0$ or $c_1$, and $\B$ checks that $(a_{\Sigma}, a, z_{\Sigma})$ is an accepting conversation. Assuming that the $\Sigma$-protocol is honest verifier zero-knowledge and $\pkB$ leads to unconditionally binding commitments, the new commitment construction is again unconditionally binding. 

During simulation, $\dhA$ chooses a $\pkB$ in the CRS such that it knows the matching decryption key $\sk$. Then, it can extract $\dA$'s choice bit $a$ by decrypting both $c_0$ and $c_1$ and checking which contains a valid $z_\Sigma$ such that $(a_{\Sigma}, a, z_{\Sigma})$ is accepting. Note that not both $c_0$ and $c_1$ can contain a valid reply, since otherwise, $\dA$ would know a $w'$ such that $(x',w') \in R$. In order to simulate in case of $\dB$, $\ddhB$ chooses the CRS as $\pkH$ and $x$ (where $x$ is such that there exists a $w$ with $(x,w) \in R$). Hence, the commitment is unconditionally hiding. Furthermore, it can be equivocated, since $\exists \ w$ with $(x,w) \in R$ and therefore, $c_0, c_1$ can both be computed with valid replies, i.e.\ $c_0 = com_{\pkH}(z_{0 \Sigma}, r)$ and $c_1 = com_{\pkH}(z_{1 \Sigma}, r')$. Quantum-computational security against $\dB$ follows from the indistinguishability of the keys $\pkB$ and $\pkH$ and the indistinguishablity of the instances $x$ and $x'$, and efficiency of both simulations is ensured due to extraction and equivocability.

\section*{Acknowledgments}
We thank Christian Schaffner and Serge Fehr for useful comments on an earlier version of the paper and the former also for discussing the issue of efficient simulation in earlier work. CL acknowledges financial support by the MOBISEQ research project funded by NABIIT, Denmark.

\bibliographystyle{alpha}
\bibliography{crypto,qip,procs}

\end{document}